\documentclass{article}
\let\svthefootnote\thefootnote
\usepackage{graphicx,spconf}
\usepackage{algorithm}% http://ctan.org/pkg/algorithms
\usepackage{algpseudocode}%
\usepackage{enumerate}
\pdfoutput=1

\usepackage{amsthm}
\makeatletter
\def\th@newremark{\th@remark\thm@headfont{\bfseries}}
\makeatletter
\usepackage[cmex10]{amsmath}
 {  

		\newtheorem{theorem}{Theorem}
    \newtheorem{proposition}{Proposition}
}

\theoremstyle{newremark}
\newtheorem{cor}{Corollary}

\theoremstyle{newremark}

\theoremstyle{newremark}
\newtheorem{fact}{Fact}

\theoremstyle{newremark}
\newtheorem{problem}{Problem}

\theoremstyle{newremark}
\newtheorem{assumption}{Assumption}

\theoremstyle{newremark}
\newtheorem{lemma}{Lemma}

\theoremstyle{definition}
\newtheorem{defn}{Definition}

\theoremstyle{newremark}

\usepackage{amssymb}

\title{A robust machine learning method for cell-load approximation in wireless networks}
\name{Daniyal Amir Awan$^{\star}$ \qquad Renato L.G. Cavalcante$^{\star\dagger}$ \qquad Slawomir Stanczak$^{\star\dagger}$}

\address{$^{\star}$ Technical University Of Berlin, Einsteinufer 27, 10587 Berlin, Germany \\
        $^{\dagger}$Fraunhofer Institute for Telecommunications, HHI, Einsteinufer 37, 10587 Berlin, Germany}
\begin{document}

%\ninept
%

\maketitle
\begin{abstract}
We propose a learning algorithm for the problem of cell-load estimation in dense wireless networks. The proposed algorithm is robust in the sense that it is designed to cope with the uncertainty due to a small number of training samples. This scenario is highly relevant in wireless networks where training has to be performed on short time scales due to fast varying communication environment. The first part of this work studies the set of achievable rates and shows that this set is compact. We then prove that the mapping relating a feasible rate vector to the corresponding unique fixed point of the non-linear load mapping is monotone and uniformly continuous. Utilizing these properties, we apply an approximation framework which achieves the best worst-case performance. Moreover, the approximation preserves the monotonicity and continuity properties. Simulation results show that the proposed method exhibits better robustness and accuracy for small training sets when compared to standard approximation techniques for multivariate data.           
\end{abstract}
\section{Introduction}
\label{sec:introduction}
The load-coupling model (see, e.g., \cite{Siomina2012,Fehske2012} and \cite{Majewski2010}) is widely used when designing networks according to the long-term evolution (LTE) standard and has also attracted attention in the context of fifth-generation (5G) networks. More specifically, the load-coupling model has been used in various optimization frameworks dealing with different aspects of network design including, but not limited to, data offloading \cite{Ho2014}, proportional fairness \cite{miguel2016}, energy optimization \cite{daniyal2016,Pollakis2016}, and load balancing \cite{siomina2012b}. 

\let\thefootnote\relax\footnote{This research was supported by Grant STA 864/9-1 from German Re-
search Foundation (DFG).}
\addtocounter{footnote}{-1}\let\thefootnote\svthefootnote

%If inter-cell interference coordination is performed, in the shape of assigning orthognal RBs to users in different cells, it will require large amount of time-frequency resources for signaling. An alternate and practical approach is to keep inter-cell coordination to a minimum and instead assign RBs randomly. The load-coupling model (see, e.g., \cite{Siomina2012,Fehske2012} and \cite{Majewski2010}) is based on such an assignemnt of RBs. 
The radio resource management (RRM) in future 5G networks is expected to be similar to the RRM in orthogonal frequency-division multiple access (OFDMA)-based networks such as LTE. Unfortunately many of the RRM problem formulations, such as small-scale optimal assignment of time-frequency resource blocks (RBs) to users, have been shown to be NP-hard \cite{Wong04}. As a result, interference models that are able to capture the long-term behavior of OFDMA-like networks while giving rise to tractable problem formulations have been the focus of recent research. The aforementioned non-linear load-coupling model is such a network-layer model that considers long-term average RB consumption and it has been shown to be sufficiently accurate \cite{Fehske2012}. In this model, the cell-load  at a base station (BS) is the proportion of RBs scheduled to support a particular rate demand. Therefore, given some power budget that can be used for transmission, each BS needs to calculate the cell-load required to serve given rate demands.
%On one hand, the value of cell-load is required by a BS to calculate the correct amount of RBs to serve a rate demand, and on the other hand, the value of load can also be seen as probability of interference caused on all RBs to other base stations (BSs) \cite{Ho2014}. 
%If BSs are operating at maximum load (by consuming all its RBs in transmission), the system cannot admit more users or tolerate fluctuation in rate demand or channel qualities. Therefore, as proposed in \cite{renato2017}, it is useful to keep low load values at all BSs to avoid user blocking or complex frequent handovers. 

In \cite{Ho2014}, the authors presented an intuitive result showing that cell-load is monotonic in user rate demand and the non-linear coupling between cells implies that increasing an arbitrary rate demand in the network increases the cell-load at each BS. Therefore before serving a higher rate demand, it is important for a BS to have a reliable estimation of the impact of this increase to the neighboring BSs in terms of cell-load and interference. This estimation can be used to make RRM and self-organizing-network (SON) algorithms more reliable and efficient. 
%Moreover for a BS, knowing the load values induced at BSs in the neighborhood before it increases the rate demand of users connected to it can be very useful. If increasing the rate takes the values of load closer to $1$, a BS might decide against it. 
The difficulty in performing such management tasks lies in the need for calculating the expected value of induced cell-load at BSs for given user rates. This is because such a computation typically uses iterative methods requiring a large amount of network information such as pathlosses, user rates etc. We provide a robust and optimal machine learning technique which allows BSs to approximate cell-load values induced for any given rate demand vector. Moreover, the complexity of the proposed method is low and the algorithm can be implemented in parallel at each BS. 
%In \cite{Ho2014}, the authors presented the intuitive result that the load is monotonic in user rate demand. If the rate demand in the network exceeds a certain value, the network might not be able to support such a demand. On the other hand, there is a close relation between network topology, transmission power values, rate and the load induced by these parameters. Therefore, in this context, the expected or average value of load for a given network parameters is an accurate indicator of whether the current configuration is feasible or not \cite{Siomina2012}. On the other hand, owing to the frequency reuse in dense networks, the value of load can also be seen as the probability of interference caused on all RBs to other base stations (BSs) \cite{Ho2014}. This observation is based upon the fact that a feasible value of load is always between $0$ and $1$. If a BS consumes almost all RBs for transmission in its cell (which results in a load value closer to $1$) it is highly likely to cause interference to neighboring cells where the frequency spectrum is being reused. 

The contributions of this study are as follows. We first study the feasible rate region and obtain some properties of the cell-load as a function of rate demand vector. The feasible rate region is defined as the set of rate demand vectors for which the cell-load at each BS is less than or equal to 1 \cite{Ho2014}. To the best of our knowledge, not much attention has been paid so far to the structure of the feasible rate region and this paper provides some initial results. In particular, we obtain the result that the feasible rate region is compact. Moreover, we prove that the cell-load is a uniformly continuous mapping over the set of feasible rate region. 
%Besides the aforementioned monotonicity result and various optimization frameworks (see, e.g., \cite{siomina2015} and \cite{Ho2014}), to the best of our knowledge cell-load as a function of rate demand has not been studied. 

Based on these results, as our second contribution, we address the problem of cell-load approximation for a given rate demand vector. Previous studies dealing with this problem, for example, in the context of data offloading \cite{Ho2014} and maximizing the scaling-up factor of traffic demand \cite{siomina2015} have used model based methods that require full information about channel gains, pathlosses etc. In contrast, we approach this problem from a machine learning perspective, in which case no channel information is required. Owing to the dynamicity of dense wireless networks, any machine learning algorithm has to train the network within a relatively small time period. As a consequence, the training sample set is small and the information about the unknown function to be approximated is scarce. In this highly relevant setting, we propose a learning algorithm which achieves best approximation in the minimax sense, i.e., where the worst possible error under uncertainty is minimized. 
Yet another difficulty lies in obtaining a shape preserving approximation. It is known that the cell-load based on the underlying load-coupling model is monotonic in rate demand \cite{Ho2014}. Therefore, the approximation should preserve this monotonicity. However, incorporating monotonicity in machine learning algorithms for multivariate data with arbitrary dimensions is difficult, and most of the well-known machine learning algorithms do not preserve the shape of the true function \cite{Beliakov2005}. The work in \cite{Kotlowski16} shows that monotonicity is also hard to incorporate in popular online learning methods.  
%For a survey on the failure of popular online schemes under monotonicity see \cite{Kotlowski16}.
In contrast to these studies, the author in \cite{Beliakov2005} proposed a shape preserving multivariate approximation. We propose a vector-valued version of this method for cell-load approximation at multiple BSs in parallel. Finally, we compare our method with popular multivariate machine learning techniques through simulations and show that our method outperforms them under the aforementioned restriction of a small training set.    
%\noindent{\bf Copyright notice 1:}\\
%For papers in which all authors are employed by the US government, the copyright notice is:\\
%{U.S.\ Government work not protected by U.S.\ copyright}
%\\[2ex]
%\noindent{\bf Copyright notice 2:}\\
%For papers in which all authors are employed by a Crown government (UK, Canada, and Australia), the copyright notice is:\\
%{978-1-5090-6341-3/17/\$31.00 {\copyright}2017 Crown}
%\\[2ex]
%\noindent{\bf Copyright notice 3:}\\
%For papers in which all authors are employed by the European Union, the copyright notice is:\\
%{978-1-5090-6341-3/17/\$31.00 {\copyright}2017 European Union}
%\\[2ex]
%\noindent{\bf Copyright notice 4:}\\
%For all other papers the copyright notice is:\\
%{978-1-5090-6341-3/17/\$31.00 {\copyright}2017 IEEE}
\section{Preliminaries}\label{sec:preliminaries}
Throughout this study, $\mathbb{R}_{+}$ and $\mathbb{R}_{++}$ denote the set of non-negative and positive reals, respectively.  We denote by $C(\mathcal{X})$ the space of vector-valued continuous functions defined on $\mathcal{X} \subset \mathbb{R}^{N}_{++}$. Likewise, we denote by $\mathbf{g}\in C(\mathcal{X})$ a vector-valued function whose values at a point $\mathbf{x}\in \mathcal{X}$ are given by $\mathbf{g}(\mathbf{x})=[g_1(\mathbf{x}),g_2(\mathbf{x}),...,g_M(\mathbf{x})]^{T}$, where each $g_i:\mathcal{X}\rightarrow [0,1]$, $i=1,\ldots,M$, is a continuous function. The norms $\left\|\cdot\right\|$ and $\left\|\cdot\right\|_{\infty}$ 
are the standard Euclidean norm and the $l_{\infty}$ norm, respectively. We denote by $\overline{\mathcal{X}}$ the closure of the set $\mathcal{X}$.
For a compact set $\mathcal{X}$ and a vector-valued continuous function $\mathbf{g}\in C(X)$ , we define the supremum or uniform norm $\left\|\cdot\right\|_{C(\mathcal{X})}$ as
\begin{equation}
\left\|\mathbf{g}\right\|_{C(\mathcal{X})}=\sup_{\mathbf{x}\in\mathcal{X}}\max_{1 \leq i \leq M}\,g_{i}(\mathbf{x}),
\label{eqn:cheb_norm}
\end{equation}
where the $\sup$ is attained because pointwise maximum of finitely many continuous functions is continuous and $\mathcal{X}$ is compact. We denote by $(\mathbf{x})_{+}$ the operation $\max\left\{\mathbf{x},\mathbf{0}\right\}$ for a vector $\mathbf{x} \in \mathbb{R}^{N}$, where in contrast to the $\max$ operation in \eqref{eqn:cheb_norm}, the $\max$ is taken component-wise and $\mathbf{0}$ is the all-zero vector. The distinction between the two usages of the $\max$ operation will be clear by the context in which they are used in the study. Finally, for two vectors $\mathbf{x}$ and $\mathbf{y}$, $\mathbf{x}\leq\mathbf{y}$ should be understood component-wise. 
%
%\begin{defn}[\textit{Lipschitz Continuity}]
%\textit{A function $f:\mathbb{R}_{++}^{N}\rightarrow [0,1]$ is called Lipschitz or Lipschitz continuous on a set $\mathcal{X}\subset\mathbb{R}_{++}^{N}$ if there exists $L\geq 0$, such that $(\forall \mathbf{x} \in \mathcal{X})(\forall \mathbf{y} \in \mathcal{X})\left|f(\mathbf{x})-f(\mathbf{y})\right|\leq L\left\|\mathbf{x}-\mathbf{y}\right\|$, and the smallest such value $L\geq 0$ is called the Lipschitz constant of $f$.  
%}\label{def:lf}
%\end{defn}
%
\begin{defn}[\textit{Monotone Function}]
\textit{Let $\mathbf{f}:\mathbb{R}_{++}^{N}\rightarrow[0,1]^{M}$ be a vector-valued function. Then $\mathbf{f}$ is called monotone on $\mathcal{X}\subset\mathbb{R}_{++}^{N}$ if $(\forall \mathbf{x} \in \mathcal{X})(\forall \mathbf{y} \in \mathcal{X}) \mathbf{x}\leq\mathbf{y}\Rightarrow\mathbf{f}(\mathbf{x})\leq\mathbf{f}(\mathbf{y})$.}
\label{def:mf}
\end{defn}
\begin{defn}[\textit{Lipschitz Monotone Functions}]
\textit{Let $\mathbf{f}:\mathbb{R}_{++}^{N}\rightarrow[0,1]^{M}$ be a vector-valued function with the $i$th component $f_i:\mathbb{R}_{++}^{N}\rightarrow [0,1], i=1,...,M$. We say that $\mathbf{f}$ belongs to the class of \textit{ Lipschitz Monotone Functions} (LIMF) if $\mathbf{f}$ is monotone on $\mathcal{X}$ and each component $f_i$ is Lipschitz on $\mathcal{X} \subset\mathbb{R}_{++}^{N}$ , i.e., $(\forall i\in\left\{1,2,\ldots,M\right\})$$(\exists L_{i}\in\mathbb{R}_{+})$$(\forall \mathbf{x} \in \mathcal{X})(\forall \mathbf{y} \in \mathcal{X})\left|f_{i}(\mathbf{x})-f_{i}(\mathbf{y})\right|\leq L_{i}\left\|\mathbf{x}-\mathbf{y}\right\|$.}
\label{def:lmf}
\end{defn}
%%SYSTEM MODEL
\section{System Model}\label{sec:system_model}
\subsection{Non-linear Load Coupling Model} \label{load_coupling_model} 
In this study, we consider a dense urban cellular base station (BS) deployment. The service area is represented by a grid of pixels, each occupying a small area which we refer to as a test point (TP) (see \cite{Siomina2012,siomina2015}, and \cite{Pollakis2016}). The concept of test points provides a network-layer view of quality-of-service (QoS) and large-scale channel conditions in a network. Within each TP, users are assumed to experience uniform signal propagation conditions. We use $r_{j}>0$ to denote the aggregated user rate demand within TP $j$ per unit time. It is assumed that if the rate requirement of each TP is met, then the QoS requirements of all users in the network are also met. Under this assumption, small-scale fluctuations in individual user rates are averaged out and are compensated by the lower layers of the protocol stack \cite{Pollakis2016}. 
We use $\mathcal{M}=\left\{1,2,...,M\right\}$ and $\mathcal{N}=\left\{1,2,...,N\right\}$ to denote the set of BSs and TPs, respectively. We consider a downlink transmission scenario and denote the vector of power levels of all BSs by $\mathbf{p}\in\mathbb{R}_{++}^M$. Throughout this study, the power and user assignment to BSs is assumed to be fixed. 
%\footnote{This is commonly the case in mobile networks because of hardware limitations and lack of power control mechanisms.}. 
We collect the rate demand of TPs in a vector $\mathbf{r}=[r_1,r_2,...,r_N]^{T}$. 
Assume a case where BS $i\in\mathcal{M}$ is serving a user in TP $j\in\mathcal{N}(i)$, where $\mathcal{N}(i)$ is the set of TPs connected to BS $i$. The bandwidth resources available at each BS are divided into resource blocks (RBs). The load-based SINR model (see \cite{Siomina2012}, \cite{Ho2014}) incorporates the  inter-cell interference from a particular BS $k \in \mathcal{M}$ as the product $p_{k}G_{k,j}\rho_{k}\geq 0$, where $G_{k,j}\geq0$ is the channel gain between BS $k$ and TP $j$ and $0\leq\rho_{k}\leq 1$ is referred to as the \textit{activity level} or \textit{cell-load} at BS $k$ \cite{Fehske2012}. With this model in hand, for given power levels $\mathbf{p}$, cell-load $\boldsymbol{\rho}$ and user assignment, the SINR $\gamma_{ij}$ of the wireless link between BS $i\in\mathcal{M}$ and TP $j\in\mathcal{N}$ is expressed as \cite{Fehske2012,Siomina2012} and \cite{Ho2014} 
\begin{equation}
\gamma_{ij}(\mathbf{p},\boldsymbol{\rho})=\frac{p_{i}G_{i,j}}{\sum_{k \in \mathcal{M}\backslash\left\{i\right\}}p_{k}G_{k,j}\rho_{k}+\sigma^{2}},
\label{eqn:SINR}
\end{equation}
where $\boldsymbol{\rho}=[\rho_1,\rho_2,...,\rho_M]^{T}$ is the vector containing the cell-load levels at all BSs in the network and $\sigma^{2}$ is the noise power.  
Given the value of SINR $\gamma_{ij}(\mathbf{p},\boldsymbol{\rho})$, a BS can transmit at rate $r_{ij}^s=B\log(1+\gamma_{ij}(\mathbf{p},\boldsymbol{\rho}))$ bits/s/RB reliably to a user in TP $j$, where $B$ is the bandwidth of the RB. Therefore, BS $i$ allocates $\rho_{ij}=r_{j}/r_{ij}^s$ RBs to meet the demand $r_{j}$. Summing over all TPs connected to BS $i$ we obtain
\begin{equation}
\rho_{i}=\frac{1}{RB}\sum_{j\in\mathcal{N}(i)} \frac{r_{j}}{\log(1+\gamma_{ij}(\mathbf{p},\boldsymbol{\rho}))},
\label{eqn:load_i}
\end{equation}  
where $R$ is the number of RBs available at each BS. For a fixed rate vector $\mathbf{r}\in\mathbb{R}^{N}_{++}$, writing the above equation for each $i\in\mathcal{M}$ results in a system of non-linear equations,
\begin{equation}
\boldsymbol{\rho}=\mathbf{q}(\boldsymbol{\rho},\mathbf{r}), 
\label{eqn:load_mapping}
\end{equation} 
where $\mathbf{q}:\mathbb{R}_{+}^{M}\times\mathbb{R}_{++}^{N}\rightarrow\mathbb{R}_{++}^{M}$ is referred to as \textit{load mapping}. Given $\mathbf{r}\in \mathbb{R}_{++}^{N}$ the mapping $\Gamma_{\mathbf{r}}:\mathbb{R}_{+}^{M}\rightarrow\mathbb{R}_{++}^{M}:\boldsymbol{\rho}\mapsto\mathbf{q}(\boldsymbol{\rho},\mathbf{r})$ is a \textit{positive concave mapping} \cite{Cavalcante2016}, so it also belongs to the class of \textit{standard interference mappings} \cite{yates95}. Therefore, for a given rate demand vector $\mathbf{r}\in \mathbb{R}_{++}^{N}$, the set $\text{Fix}(\Gamma_{\mathbf{r}}):=\left\{\boldsymbol{\rho}\in\mathbb{R}^{M}_{++}|\Gamma_{\mathbf{r}}(\boldsymbol{\rho})=\boldsymbol{\rho}\right\}$ contains at most one fixed point. If $\text{Fix}(\Gamma_{\mathbf{r}})\neq\emptyset$, the unique fixed point is the solution to \eqref{eqn:load_mapping}. We define a \textit{feasible} rate demand as follows: 
\begin{defn}[\textit{Feasible Rate Demand Vector}]
\textit{A rate demand vector $\mathbf{r}\in\mathbb{R}_{++}^{N}$ is \textit{feasible} for the network if and only if $\text{Fix}(\Gamma_{\mathbf{r}})\neq\emptyset$ and
$\boldsymbol{0}\leq \boldsymbol{\rho}^{\ast}\leq\mathbf{1}$, where $\boldsymbol{\rho}^{\ast}\in\text{Fix}(\Gamma_{\mathbf{r}})$}. 
\label{defn:feasible_set}
\end{defn}

Denote by $\mathcal{X}_{f}$ the set of all feasible rate demand vectors as defined in Definition \ref{defn:feasible_set}. Denote by $\mathbb{R}^{N}_{++}\ni\mathbf{r}_{\text{min}}$, the minimum rate requirement of users and consider the set $\mathcal{X}_{\text{min}}:=\{\mathbf{r} \in \mathbb{R}_{++}^{N}|\mathbf{r}\geq \mathbf{r}_{\text{min}}\}$. For the remainder we define the set of feasible rate demand vectors as $\mathcal{X}:=\mathcal{X}_{f}\bigcap \mathcal{X}_{\text{min}}$ and the set of fixed points as $\mathcal{Y}:=\left\{\boldsymbol{\rho}\in [0,1]^M|(\exists\mathbf{r}\in\mathcal{X})\Gamma_{\mathbf{r}}(\boldsymbol{\rho})=\boldsymbol{\rho}\right\}$. Furthermore, we assume that $\mathcal{X},\mathcal{Y}\neq\emptyset$. 
%Furthermore, we assume that $\exists \mathbf{r}_{\text{min}}\in\mathcal{X}$ and $\exists \epsilon >0$ such that $(\forall \mathbf{r} \in \mathcal{X})$ $\mathbf{r}\geq \mathbf{r}_{\text{min}} > \epsilon.\mathbf{1}$. Note that these assumption are easily satisfied in practical networks with per user minimum rate requirements.  
%

In the following section, we proceed to study some properties of the feasible set of rate vectors and the corresponding fixed points. For ease of reference, we first present some important properties of positive concave mappings.
\begin{fact}\cite{yates95}
Let $\Gamma:\mathbb{R}_{+}^{M}\rightarrow\mathbb{R}_{++}^{M}$ be a positive concave mapping. Then each of the following holds:
\begin{enumerate}
\item $\text{Fix}(\Gamma)\neq\emptyset$ if and only if there exists $\boldsymbol{\rho}\in\mathbb{R}_{+}^{M}$ such that $\Gamma(\boldsymbol{\rho})\leq\boldsymbol{\rho}$.
\item If it exists, the fixed point can be found by the fixed point iteration $\boldsymbol{\rho}_{n+1}=\Gamma(\boldsymbol{\rho}_{n})$, $n\in\mathbb{N}$, with $\boldsymbol{\rho}_{1}\in\mathbb{R}_{+}^{M}$ chosen arbitrarily.
\item The fixed point iteration is increasing (resp. decreasing) if $\boldsymbol{\rho}_{1}\leq\Gamma(\boldsymbol{\rho}_{1}) (\text{resp.} \boldsymbol{\rho}_{1}\geq\Gamma(\boldsymbol{\rho}_{1}))$.
\end{enumerate}
\label{rem:fixed_point}
\end{fact}
%%Feasible Rate Demand Set
\section{Feasible Rate Region and Fixed Points}\label{sec:feasible_rate_demand_set}
In this section, we show that the feasible rate region is compact and the fixed points are generated by a uniformly continuous monotonic mapping on this set. The compactness of the domain set and continuity of the function to be approximated are two important conditions in theory of minimax optimal recovery/approximation of functions (\cite{Golub1971}), which we use in the proposed learning algorithm in Section \ref{section:the_learning_problem}.

We start this section by two simple results. The first result shows that the fixed point of the load mapping in \eqref{eqn:load_mapping} is monotonic in rate demand. See \cite[Theorem 2]{Ho2014} for an alternative proof. We provide a simpler proof that suffices for our purposes. 
\begin{lemma}
\textit{Consider any two rate demand vectors $\mathbf{r}^{k},\mathbf{r}^{j} \in \mathcal{X}$ and the fixed points $\boldsymbol{\rho}^j\in\text{Fix}(\Gamma_{\mathbf{r}^{j}})$ and $\boldsymbol{\rho}^k\in\text{Fix}(\Gamma_{\mathbf{r}^{k}})$. Then $\mathbf{r}^{j} \geq \mathbf{r}^{k} \Rightarrow \boldsymbol{\rho}^j\geq\boldsymbol{\rho}^k$.}
\label{lemma:one}
\end{lemma}

\begin{proof}Assume that $\boldsymbol{\rho}^k=\mathbf{q}(\boldsymbol{\rho}^k,\mathbf{r}^k)=\Gamma_{\mathbf{r}^k}(\boldsymbol{\rho}^k)$ and consider a demand vector $\mathbf{r}^j\geq\mathbf{r}^k$. Then we have from \eqref{eqn:load_mapping} that $\boldsymbol{\rho}^k=\mathbf{q}(\boldsymbol{\rho}^k,\mathbf{r}^k)\leq\mathbf{q}(\boldsymbol{\rho}^k,\mathbf{r}^j)$. Starting with the vector $\boldsymbol{\rho}^k$, we employ the standard fixed point iteration (see Fact \ref{rem:fixed_point}) and obtain a monotonically increasing sequence $(\boldsymbol{\rho}_{n+1}=\mathbf{q}(\boldsymbol{\rho}_{n},\mathbf{r}^j))_{n\in\mathbb{N}}$, with $\boldsymbol{\rho}_{1}=\boldsymbol{\rho}^k$, which converges to the uniquely existing fixed point $\boldsymbol{\rho}^j=\mathbf{q}(\boldsymbol{\rho}^j,\mathbf{r}^j)=\Gamma_{\mathbf{r}^j}(\boldsymbol{\rho}^j)\geq\boldsymbol{\rho}^k$.
\end{proof}
%
%\begin{corollary}
%An immediate consequence of Lemma \ref{lemma:one} is that if $\exists\mathbf{r}^{\text{min}}\in\mathcal{X}$ such that $(\forall\mathbf{r}\in\mathcal{X})$ $\mathbf{r}^{\text{min}}\leq\mathbf{r}$, then $(\forall\mathbf{r}\in\mathcal{X})$ $\Gamma_{\mathbf{r}^{\text{min}}}(\boldsymbol{\rho}^{\text{min}})=\boldsymbol{\rho}^{\text{min}} \leq \boldsymbol{\rho}=\Gamma_{\mathbf{r}}(\boldsymbol{\rho})\leq \mathbf{1}$ and the set of fixed points is given by $\mathcal{Y}=\left\{\boldsymbol{\rho}\in\mathbb{R}^{M}_{++}|(\forall\mathbf{r}\in\mathcal{X})\boldsymbol{\rho}^{\text{min}}\leq\boldsymbol{\rho}=\Gamma_{\mathbf{r}}(\boldsymbol{\rho})\leq \mathbf{1}\right\}$.
%\end{corollary}
%

\begin{lemma}
The feasible rate region $\mathcal{X}$ is bounded.
\label{lemma:two}
\end{lemma}

\begin{proof}
The set $\mathcal{X}$ is clearly bounded from below. Suppose the set $\mathcal{X}$ is unbounded from above, then there exists at least one unbounded sequence $(\mathbf{r}_{n})_{n\in\mathbb{N}}\subset\mathcal{X}$. This implies that at least one component of the vector $\mathbf{r}_{n}$ grows unbounded. Let us denote the sequence of this component by $(r_{l,n})_{n\in\mathbb{N}}$ and let the corresponding TP be  connected to BS $i\in\mathcal{M}$. Every unbounded sequence has an increasing subsequence that diverges to $+ \infty$. Let us extract such a subsequence and denote  it by $(r_{l,k})_{k\in\mathbb{K}\subset\mathbb{N}}$. Likewise, denote by $(\rho_{i,k})_{k\in\mathbb{K}\subset\mathbb{N}}$ the $i$th component of the subsequence $(\boldsymbol{\rho}_{k})_{k\in\mathbb{K}\subset\mathbb{N}}$, where $\boldsymbol{\rho}_k \in \text{Fix}(\Gamma_{\boldsymbol{r}_k})$. It can be easily verified that, for fixed $\mathbf{p}$ and bandwidth resources (RB) in $\eqref{eqn:load_i}$, we have that $\rho_{i,k}=\frac{1}{RB}\sum_{j\in\mathcal{N}(i)} \frac{r_{j,k}}{\log(1+\gamma_{ij}(\mathbf{p},\boldsymbol{\rho}_k))}\geq\frac{1}{RB}\sum_{j\in\mathcal{N}(i)}\frac{r_{j,k}}{\log(1+\gamma_{ij}(\mathbf{p},\boldsymbol{0}))}\geq\frac{1}{RB}\frac{r_{l,k}}{\log(1+\gamma_{ij}(\mathbf{p},\boldsymbol{0}))}>0$, where $\rho_{i,k}$ and $r_{j,k}$ are the $i$th and $j$th component of vectors $\boldsymbol{\rho}_k$ and $\mathbf{r}_{k}$, respectively. Now, note that the lower bound $\frac{1}{RB}\frac{r_{l,k}}{\log(1+\gamma_{ij}(\mathbf{p},\boldsymbol{0}))}$ grows unbounded as $r_{l,k}\rightarrow\infty$, which in particular implies that $\lim_{k\rightarrow\infty}\rho_{i,k}=\infty$. However, this contradicts the fact that, by our definition of feasibility,  $(\forall k\in\mathbb{K})$ $0\leq\rho_{i,k}\leq 1$. Therefore, we conclude that the feasible set $\mathcal{X}$ does not contain any unbounded sequence, so $\mathcal{X}$ is bounded.
\end{proof}

Now we proceed to characterize the fixed point solution of \eqref{eqn:load_mapping} by exploiting its connection with the conditional eigenvalue problem based on concave Perron-Frobenius theory \cite{krause01}. First, we prove a general result which is valid for positive concave mappings\footnote{The same result also holds for a more general class of standard interference mappings but we do not consider these in this study \cite[Proposition~2]{renato20172}.} arising in various contexts beyond the scope of this paper. 
\begin{proposition}
\textit{ Let $\Gamma:\mathbb{R}_{+}^{M}\rightarrow\mathbb{R}_{++}^{M}$ be an arbitrary positive concave mapping and
consider the following conditional eigenvalue problem which always has a unique solution $(\boldsymbol{\rho}^{\ast},\lambda^{\ast})\in\mathbb{R}^{M}_{++} \times \mathbb{R}_{++}$ \cite{krause01}:}
\begin{problem}
\textit{Find $(\boldsymbol{\rho}^{\ast},\lambda^{\ast})\in\mathbb{R}^{M}_{+} \times \mathbb{R}_{+}$ such that $\Gamma(\boldsymbol{\rho}^{\ast})=\lambda^{\ast}\boldsymbol{\rho}^{\ast}$ and $\left\|\boldsymbol{\rho}^{\ast}\right\|_{\infty}=1$.} 
\label{problem:eig1}
\end{problem}
\textit{Then $\text{Fix}(\Gamma)\neq\emptyset$ and $\boldsymbol{1}\geq\boldsymbol{\rho}'\in \text{Fix}(\Gamma)$ if and only if $0 <\lambda^{\ast}\leq 1$.} 
\label{proposition:one}
\end{proposition}
\begin{proof}
First note that $0 <\lambda^{\ast}$ since $(\forall \boldsymbol{\rho}\in\mathbb{R}_{+}^{M})$ $\Gamma(\boldsymbol{\rho}) \in \mathbb{R}^{M}_{++}$. If $\lambda^{\ast}\leq 1$, then $\Gamma(\boldsymbol{\rho}^{\ast})\leq\boldsymbol{\rho}^{\ast}$ and, by Fact 1.1, it follows that $\text{Fix}(\Gamma)\neq\emptyset$ and the fixed point $\boldsymbol{\rho}'\in \text{Fix}(\Gamma)$ exists and satisfies $\boldsymbol{0} \leq\boldsymbol{\rho}'\leq \boldsymbol{\rho}^{\ast}\leq \boldsymbol{1}$. This shows one direction of the claim.

Conversely, assume $\boldsymbol{\rho}'\in \text{Fix}(\Gamma)$ and $\boldsymbol{0}\leq \boldsymbol{\rho}' \leq \boldsymbol{1}$. Let $\left\|\boldsymbol{\rho}'\right\|_{\infty}=c$, where $0 < c \leq 1$, and $\lambda'=1$. Then $(\boldsymbol{\rho}',\lambda')$ solves the following conditional eigenvalue problem:
%First note that by definition $\Gamma_{\mathbf{r}}(\boldsymbol{\rho}^{\ast})=\lambda^{\ast}\boldsymbol{\rho}^{\ast}$ $\Leftrightarrow$ $\mathbf{q}(\boldsymbol{\rho}^{\ast},\frac{1}{\lambda^{\ast}}\mathbf{r})=\boldsymbol{\rho}^{\ast}$. Suppose $\lambda^{\ast}\leq 1$, then we have that $\frac{1}{\lambda^{\ast}}\mathbf{r}\geq\mathbf{r}$. Let $\overline{\mathbf{r}}=\frac{1}{\lambda^{\ast}}\mathbf{r}$, then it can be seen that $\boldsymbol{\rho}^{\ast}=\mathbf{q}(\boldsymbol{\rho}^{\ast},\overline{\mathbf{r}})$ and $\boldsymbol{\rho}^{\ast}$ is the fixed point of the mapping $\Gamma_{\overline{\mathbf{r}}}$ with $\left\|\boldsymbol{\rho}^{\ast}\right\|_{\infty}=1$. Observing that $\Gamma_{\mathbf{r}}(\boldsymbol{\rho}^{\ast})\leq\Gamma_{\overline{\mathbf{r}}}(\boldsymbol{\rho}^{\ast})=\boldsymbol{\rho}^{\ast}$, it follows from Fact \ref{rem:fixed_point} that $\boldsymbol{0}\leq \text{Fix}(\Gamma_{\mathbf{r}})\leq \boldsymbol{\rho}^{\ast} \leq \boldsymbol{1}$. Conversely, denote the fixed point of $\Gamma_{\mathbf{r}}$ by $\boldsymbol{\rho}'$ and suppose that $\boldsymbol{0}\leq \boldsymbol{\rho}' \leq \boldsymbol{1}$. Then $(\boldsymbol{\rho}',1)$ solves the following eigenvalue problem:
\begin{problem}
\textit{Find $(\boldsymbol{\rho},\lambda)\in\mathbb{R}^{M}_{+} \times \mathbb{R}_{+}$ 
such that $\Gamma(\boldsymbol{\rho})=\lambda\boldsymbol{\rho}$ and $\left\|\boldsymbol{\rho}\right\|_{\infty}=c$.}
\label{problem:eig2}
\end{problem}
Now consider Problem \ref{problem:eig1} with this mapping $\Gamma$, and denote its solution by $(\boldsymbol{\rho}^{\ast},\lambda^{\ast})\in\mathbb{R}^{M}_{++} \times \mathbb{R}_{++}$. To obtain a contradiction suppose that $\lambda^{\ast}>1=\lambda'$. From \cite[Lemma 3.1.3]{nuzman07}, and the fact that positive concave mappings are a special case of standard interference mappings \cite{Cavalcante2016}, it follows that $\lambda^{\ast}>\lambda'\Rightarrow\boldsymbol{\rho}'>\boldsymbol{\rho}^{\ast}$. 
From the monotonicity of the $l_{\infty}$ norm, we obtain the contradiction that  $c=\left\|\boldsymbol{\rho}'\right\|_{\infty}>\left\|\boldsymbol{\rho}^{\ast}\right\|_{\infty}=1$. This concludes the second direction of the proof.
%, where $\left\|\boldsymbol{\rho}^{\ast}\right\|_{\infty}=c^{\ast}=1$. 
%From \cite[Lemma 3.1.3]{nuzman07}, and the fact that positive concave mappings are a special case of standard interference mappings \cite{Cavalcante2016}, it follows that $c^{\ast}\geq c' \Rightarrow \lambda^{\ast}\leq\lambda'=1$. 
%$\boldsymbol{\rho}'>\boldsymbol{\rho}^{\ast}$. From the monotonicity of the $l_{\infty}$ norm, we obtain the contradiction that  $\left\|\boldsymbol{\rho}'\right\|_{\infty}>\left\|\boldsymbol{\rho}^{\ast}\right\|_{\infty}=1$. This concludes the second direction of the proof.  
\end{proof}
%corollary
%
The immediate consequence of Proposition \ref{proposition:one}, which we will utilize in Proposition 2, is that a rate vector $\mathbf{r}\in\mathbb{R}^{N}_{++}$ is feasible if and only if the solution of the conditional eigenvalue Problem \ref{problem:eig1} with the positive concave mapping $\Gamma_{\mathbf{r}}$ satisfies $\lambda^{\ast}\leq 1$. We also obtain the following result from Proposition \ref{proposition:one}:
\begin{cor}\label{cor:correlary_one} 
\textit{Let $\lambda^{\ast}$ be the eigenvalue that solves Problem \ref{problem:eig1} in Proposition \ref{proposition:one} with the positive concave mapping $\Gamma_{\mathbf{r}}$. As a function of $\mathbf{r}\in\mathcal{X}$, $\lambda^{\ast}(\mathbf{r})$ is bounded away from 0, i.e., $(\exists \epsilon >0)$ $(\forall \mathbf{r} \in \mathcal{X})$ $\lambda^{\ast}(\mathbf{r})\geq\epsilon $}.
\end{cor}
\begin{proof} 
Consider a sequence $(\mathbf{r}_{n})_{n\in\mathbb{N}}\subset\mathcal{X}$ and sequence of solutions $(\lambda^{\ast}_{n},\boldsymbol{\rho}^{\ast}_{n})_{n\in\mathbb{N}}$ to Problem \ref{problem:eig1}, where for convenience we define $\lambda^{\ast}_{n}:=\lambda^{\ast}(\mathbf{r}_{n})$ in the proof. Now suppose that $(\forall n\in\mathbb{N})$ $0 <\lambda^{\ast}_{n}\leq 1$ and there exists a sequence $(\lambda^{\ast}_{n})_{n\in\mathbb{N}}$ such that $\lambda^{\ast}_{n}\rightarrow 0$. 
%To see this, suppose that $(\forall \epsilon \geq 0)(\exists \mathbf{r}_{\epsilon}\in\mathbb{R}^{N}_{++})$$0\leq\lambda_{\epsilon}^{\ast}\leq\epsilon$, where $\lambda_{\epsilon}^{\ast}$ is the eigenvalue solution as a function of $\mathbf{r}_{\epsilon}$. Pick a particular $\epsilon\leq 1$. 
By definition, for a given $n\in\mathbb{N}$, if $(\lambda^{\ast}_{n},\boldsymbol{\rho}^{\ast}_{n})$ solves Problem \ref{problem:eig1}, then $\Gamma_{\mathbf{r}_{n}}(\boldsymbol{\rho}^{\ast}_{n})=\lambda^{\ast}_{n}\boldsymbol{\rho}^{\ast}_{n}$ $\Leftrightarrow$ $\mathbf{q}(\boldsymbol{\rho}^{\ast}_{n},\frac{1}{\lambda_{n}^{\ast}}\mathbf{r}_{n})=\boldsymbol{\rho}^{\ast}_{n}$ and $\left\|\boldsymbol{\rho}^{\ast}_{n}\right\|_{\infty}=1$. Let $\overline{\mathbf{r}}_{n}=\frac{1}{\lambda_{n}^{\ast}}\mathbf{r}_{n}\geq \frac{1}{\lambda_{n}^{\ast}}\mathbf{r}_{\text{min}}$. It can be seen that $\boldsymbol{\rho}^{\ast}_{n}=\mathbf{q}(\boldsymbol{\rho}^{\ast}_{n},\overline{\mathbf{r}}_{n})$, so $\boldsymbol{\rho}^{\ast}_{n}$ is the fixed point of the mapping $\Gamma_{\overline{\mathbf{r}}_{n}}$, which implies that $\overline{\mathbf{r}}_{n}$ is feasible since $\boldsymbol{\rho}^{\ast}_{n}\leq\boldsymbol{1}$. But note that as $\lambda^{\ast}_{n}\rightarrow 0$, $\left\|\overline{\mathbf{r}}_{n}\right\|_{\infty}\rightarrow\infty$ which contradicts Lemma \ref{lemma:two}. So $\exists \epsilon >0$ such that $(\forall \mathbf{r} \in \mathcal{X})$ $\epsilon\leq\lambda^{\ast}(\mathbf{r})\leq 1$.
\label{corollary:one}  
\end{proof}

In what follows, we denote by $\mathbf{f}:\mathcal{X} \rightarrow \mathcal{Y}:\mathbf{r}\mapsto\boldsymbol{\rho}\in\text{Fix}(\Gamma_\mathbf{r})$ the function that maps each $\mathbf{r}$ to the unique fixed point of the mapping $\Gamma_{\mathbf{r}}$ in \eqref{eqn:load_mapping}. We are now in a position to present the main results of this section which enable us to apply robust and optimal approximation methods in the next section. 
\begin{proposition}
\textit{The feasible rate region $\mathcal{X}\subset\mathbb{R}^{N}_{++}$ is compact.}
%\left\{\mathbf{r}\in\mathbb{R}^{N}_{++}|\mathbf{q}(\boldsymbol{\rho}^{\ast},\mathbf{r})-\lambda^{\ast}\boldsymbol{\rho}^{\ast}=\boldsymbol{0},\exists\epsilon>0:0<\epsilon\leq\lambda^{\ast}\leq 1 ,\linebreak \left\|\boldsymbol{\rho}^{\ast}\right\|_{\infty}=1\right\}$. Then $\mathcal{X}$ is compact.
\label{proposition:two}
\end{proposition}
%proof
\begin{proof}
We have shown by Proposition \ref{proposition:one} that a rate vector $\mathbf{r}\in\mathcal{X}$ is feasible if and only if the solution to Problem \ref{problem:eig1}, with the positive concave mapping $\Gamma_{\mathbf{r}}$, satisfies $(\exists\epsilon>0)$ $\epsilon\leq\lambda^{\ast}\leq 1$. Now, let $\Lambda= \left\{\lambda\in\mathbb{R}_{++}|\epsilon\leq\lambda\leq 1 \right\}$ and $\mathcal{P}=\left\{\boldsymbol{\rho}\in\mathbb{R}^{M}_{+}|\left\|\boldsymbol{\rho}\right\|_{\infty}=1\right\}$. Since the set $\mathcal{X}$ is bounded, every sequence of rate demand vectors has a convergent subsequence. The same holds for bounded sequences $(\boldsymbol{\rho}_n^{\ast})_{n\in\mathbb{N}} \subset \mathcal{P}$ and $(\lambda^{\ast}_n)_{n\in\mathbb{N}} \subset \Lambda$. Since $(\mathbf{r}_n)_{n \in\mathbb{N}}$ is bounded, it has a convergent subsequence $(\mathbf{r}_n)_{n\in K_{1}\subset \mathbb{N}}$ whose every subsequence is also convergent. Denote by $(\boldsymbol{\rho}_n^{\ast})_{n\in K_{1}\subset \mathbb{N}}$ and $(\lambda^{\ast}_n)_{n\in K_{1}\subset \mathbb{N}}$ the sequences of solutions to Problem \ref{problem:eig1} corresponding to $(\mathbf{r}_n)_{n\in K_{1}\subset \mathbb{N}}$. 

Since the sequence $(\boldsymbol{\rho}_n^{\ast})_{n\in K_{1}\subset \mathbb{N}}$ is bounded, we can extract a convergent subsequence $(\boldsymbol{\rho}_n^{\ast})_{n\in K_{2}\subset K_{1}}$ whose corresponding subsequence of rate vectors $(\mathbf{r}_n)_{n\in K_{2}\subset K_{1}}$ is also convergent. Furthermore, denote by $(\lambda_n^{\ast})_{n\in K_{2}\subset K_{1}}$ the bounded sequence corresponding to $(\boldsymbol{\rho}_n^{\ast})_{n\in K_{2}\subset K_{1}}$. Since $(\lambda_n^{\ast})_{n\in K_{2}\subset K_{1}}$ is bounded, we can extract a convergent subsequence $(\lambda_n^{\ast})_{n\in K_{3}\subset K_{2}}$ whose corresponding sequences of fixed points $(\boldsymbol{\rho}_n^{\ast})_{n\in K_{3}\subset K_{2}}$ and rate vectors $(\mathbf{r}_n)_{n\in K_{3}\subset K_{2}}$ are subsequences of convergent sequences and, therefore, convergent.  

Denote by $(\overline{\boldsymbol{\rho}^{\ast}},\overline{\lambda^{\ast}})$ the limit of the subsequence $(\boldsymbol{\rho}_n^{\ast},\lambda_n^{\ast})_{n\in K_{3}}$ which belongs to the compact set $\mathcal{P} \times \Lambda$. For each $n\in K_{3}$, $\mathbf{r}_n$ is feasible if and only if $\mathbf{g}(\boldsymbol{\rho}^{\ast}_{n},\mathbf{r}_n,\lambda^{\ast}_{n})=\mathbf{q}(\boldsymbol{\rho}^{\ast}_{n},\mathbf{r}_n)-\lambda^{\ast}_{n}\boldsymbol{\rho}^{\ast}_{n}=\mathbf{0}.$ Note that since $\mathbf{g}$ is continuous, $\lim_{n \in K_{3}}\mathbf{g}(\boldsymbol{\rho}^{\ast}_{n},\mathbf{r}_n,\lambda^{\ast}_{n})=\mathbf{g}(\overline{\boldsymbol{\rho}^{\ast}},\overline{\mathbf{r}},\overline{\lambda^{\ast}})=\mathbf{0}$ which implies that $\overline{\mathbf{r}}\in\mathcal{X}$. 

Since $(\mathbf{r}_n)_{n\in\mathbb{N}}$ was chosen arbitrarily, the above holds for each sequence in $\mathcal{X}$ which implies that $\mathcal{X}$ is closed. Furthermore, in finite dimensional metric spaces every bounded and closed set is compact. 
%Let $(\boldsymbol{\rho}_n^{\ast})_{n\in\mathbb{N}}\subset\mathcal{P}$ and $(\lambda^{\ast}_n)_{n\in\mathbb{N}}\subset\Lambda$ be two convergent subsequences of arbitrary sequences of the solutions of the eigenvalue problem and denote by $\overline{\boldsymbol{\rho}^{\ast}}$ and $\overline{\lambda^{\ast}}$ their limits which exist and belong to compact sets $\mathcal{P}$ and $\Lambda$, respectively.
 %For each $n\in\mathbb{N}$, we have $\mathbf{g}(\boldsymbol{\rho}^{\ast}_{n},\mathbf{r}_n,\lambda^{\ast}_{n})=\mathbf{q}(\boldsymbol{\rho}^{\ast}_{n},\mathbf{r}_n)-\lambda^{\ast}_{n}\boldsymbol{\rho}^{\ast}_{n}=0.$ Since the function $\mathbf{g}$ is continuous, it follows that $\lim_{(\boldsymbol{\rho}_n^{\ast},\lambda^{\ast}_{n},\mathbf{r}_{n})\rightarrow(\overline{\boldsymbol{\rho}^{\ast}},\overline{\lambda^{\ast}},\overline{\mathbf{r}})}\mathbf{g}(\boldsymbol{\rho}^{\ast}_{n},\mathbf{r}_n,\lambda^{\ast}_{n})=\mathbf{g}(\overline{\boldsymbol{\rho}^{\ast}},\overline{\mathbf{r}},\overline{\lambda^{\ast}})=\mathbf{0}$ from which we conclude that $\overline{\mathbf{r}}\in\mathcal{X}$. Since $(\mathbf{r}_n)_{n\in\mathbb{N}}$ is a subsequence of an arbitrary sequence, the above holds for each sequence in $\mathcal{X}$ which implies that $\mathcal{X}$ is closed. Furthermore, in finite dimensional metric spaces every bounded and closed set is compact. 
\end{proof}

\begin{theorem}\textit{
Consider the function $\mathbf{f}:\mathcal{X} \rightarrow \overline{\mathcal{Y}}:\mathbf{r}\mapsto\boldsymbol{\rho}\in\text{Fix}(\Gamma_r)$. The function $\mathbf{f}$ is uniformly continuous over the compact set $\mathcal{X}$.
\label{proposition:three}}
\end{theorem}
\begin{proof} 
%Since the set $\mathcal{X}$ is compact, it is sufficient to show that $\mathbf{f}$ is continuous on $\mathcal{X}$. Extract a convergent sequence $(\mathbf{r}_n)_{n \in\mathbb{N}}\subset\mathcal{X}$, and let the point $\mathbf{r}^{\ast}$ be its limit which, since $\mathcal{X}$ is compact, exists and belongs to $\mathcal{X}$. It suffices to show that $\lim_{\mathbf{r}_{n}\rightarrow\mathbf{r}^{\ast}}\mathbf{f}(\mathbf{r}_{n})=\mathbf{f}(\mathbf{r}^{\ast})$. Extract a convergent sequence $(\boldsymbol{\rho}_{n})_{n \in\mathbb{N}}\subset \mathcal{Y}$, of corresponding fixed points in the compact set $\mathcal{Y}$ and denote its limit by $\boldsymbol{\rho}^{\ast}\in\mathcal{Y}$. Now consider the function $\mathbf{g}(\boldsymbol{\rho},\mathbf{r})=\boldsymbol{\rho}-\mathbf{q}(\boldsymbol{\rho},\mathbf{r})$, where $\mathbf{q}$ is the load mapping, and note that this function is continuous and  $(\forall n \in \mathbb{N})$ $\mathbf{g}(\boldsymbol{\rho}_n,\mathbf{r}_n)=\mathbf{0}$. It follows from the definition of a continuous function that $\lim_{(\boldsymbol{\rho}_{n},\mathbf{r}_{n})\rightarrow(\boldsymbol{\rho}^{\ast},\mathbf{r}^{\ast})}\mathbf{g}(\boldsymbol{\rho}_n,\mathbf{r}_n)=\mathbf{g}(\boldsymbol{\rho}^{\ast},\mathbf{r}^{\ast})=\mathbf{0}$. Therefore, the limit of the subsequence $(\boldsymbol{\rho}_{n})_{n \in\mathbb{N}}$ is the unique fixed point $\mathbf{f}(\mathbf{r}^{\ast})=\boldsymbol{\rho}^{\ast}$ and $\lim_{\mathbf{r}_{n}\rightarrow\mathbf{r}^{\ast}}\mathbf{f}(\mathbf{r}_{n})=\mathbf{f}(\mathbf{r}^{\ast})$.
Since $\mathcal{X}$ is compact, every infinite sequence in $\mathcal{X}$ has a convergent subsequence whose limit is in $\mathcal{X}$. Let $(\mathbf{r}_n)_{n \in\mathbb{N}}\subset\mathcal{X}$ be an arbitrary convergent sequence, and let the point $\mathbf{r}^{\ast}\in\mathcal{X}$ be its limit. To prove that $\mathbf{f}$ is continuous, we need to show that $\lim_{n \rightarrow \infty}\mathbf{f}(\mathbf{r}_{n})=\mathbf{f}(\mathbf{r}^{\ast})$. To this end, let $(\boldsymbol{\rho}_{n})_{n \in\mathbb{N}}\subset \overline{\mathcal{Y}}$ be the corresponding sequence of $(\mathbf{f}(\mathbf{r}_{n}))_{n \in\mathbb{N}}$. Since $\overline{\mathcal{Y}}\subset[0,1]^{M}$\footnote{The set $\mathcal{Y}$ is clearly bounded. It is also closed, i.e. the closure $\overline{\mathcal{Y}}=\mathcal{Y}$, but we omit the proof for brevity. Every closed and bounded set in finite-dimensional normed spaces is compact} is compact, such a sequence has a convergent subsequence $(\boldsymbol{\rho}_{n})_{n \in K_{1} \subset \mathbb{N}}$ whose limit $\boldsymbol{\rho}^{\ast}$ exists and belongs to $\overline{\mathcal{Y}}$. The corresponding sequence of rate vectors $(\mathbf{r}_{n})_{n \in K_{1} \subset \mathbb{N}}$ is a subsequence of the convergent sequence $(\mathbf{r}_n)_{n \in\mathbb{N}}\subset\mathcal{X}$ and therefore also convergent.

Now consider the function $\mathbf{g}(\boldsymbol{\rho},\mathbf{r})=\boldsymbol{\rho}-\mathbf{q}(\boldsymbol{\rho},\mathbf{r})$, where $\mathbf{q}$ is the load mapping, and note that this function is continuous and  $(\forall n \in K_{1} \subset \mathbb{N})$ $\mathbf{g}(\boldsymbol{\rho}_n,\mathbf{r}_n)=\mathbf{0}$. It follows from the definition of a continuous function that $\lim_{n \in K_{1}}\mathbf{g}(\boldsymbol{\rho}_n,\mathbf{r}_n)=\mathbf{g}(\boldsymbol{\rho}^{\ast},\mathbf{r}^{\ast})=\mathbf{0}$. Therefore, the limit of the subsequence $(\boldsymbol{\rho}_{n})_{\in K_{1} \subset \mathbb{N}}$ is the unique fixed point $\mathbf{f}(\mathbf{r}^{\ast})=\boldsymbol{\rho}^{\ast}$ and $\lim_{n \in K_{1}}\mathbf{f}(\mathbf{r}_{n})=\mathbf{f}(\mathbf{r}^{\ast})$. Since $\mathcal{X}$ is compact, $\mathbf{f}$ is uniformly continuous on $\mathcal{X}$.
\end{proof}
%%Problem Formulation 
\section{The Learning Problem}\label{section:the_learning_problem}
In this section we present a learning algorithm which is not only robust and optimal in a challenging machine learning scenario, but also preserves the monotonicity and continuity of the function to be approximated.  
\subsection{Minimax Optimal Approximation}
Let the training data set be denoted by $\mathcal{D}=\{(\mathbf{r}^k,\boldsymbol{\rho}^k)\}^{K}_{k=1}$, $(\mathbf{r}^k,\boldsymbol{\rho}^k)\in (\mathcal{X} \times \mathcal{Y})$, where $\boldsymbol{\rho}^k:=\mathbf{f}(\mathbf{r}^k)$ are the measured cell-load values generated by the underlying function $\mathbf{f}:\mathcal{X}\rightarrow\mathcal{Y}$. Our objective is to approximate the value $\mathbf{f}(\mathbf{r})$ for any $\mathbf{r} \in \mathcal{X}$, i.e., our objective is to solve an interpolation problem given $\mathcal{D}$.

In the classical approximation theory (see, for example, \cite{Weinberger1959,Sukharev1992,traub1980}), the data interpolation problem entails computing an approximation $\mathbf{g}$ of the function $\mathbf{f}$ by observing the values in the set $\mathcal{D}$ and then replacing future evaluations of $\mathbf{f}(\mathbf{r})$ with $\mathbf{g}(\mathbf{r})$ for any $\mathbf{r} \in \mathcal{X}$. We have shown by Theorem \ref{proposition:three} that the function $\mathbf{f}:\mathbf{r}\mapsto\boldsymbol{\rho}\in\text{Fix}(\Gamma_\mathbf{r})$ is a uniformly continuous function on the compact set $\mathcal{X}$. 
%For a fixed continuous function $\mathbf{f}\in C(\mathcal{X})$ on a compact set $\mathcal{X}$, there exists a unique interpolating polynomial function $\mathbf{g}$ of degree less or equal to $K$ that gives the best (minimum norm-distance) approximation of $\mathbf{f}$ given the dataset $\mathcal{D}$ (reference missing). However, 
Clearly there are infinitely many functions in the space $C(\mathcal{X})$ that interpolate $\mathcal{D}$. Since we are interested in a robust approximation of the unknown $\mathbf{f}^{\ast} \in C(\mathcal{X})$, we aim at minimizing the 
worst-case error \cite{Sukharev1992},
\begin{equation}
\text{E}_{w}(\mathbf{g})=\underset{\mathbf{f}\in C(\mathcal{X})}{\sup}\left\|\mathbf{f}-\mathbf{g}\right\|_{C(\mathcal{X})},
\label{eq:error}
\end{equation}
where $\mathbf{g}\in C(\mathcal{X})$ is confined to a class of functions such that $\mathbf{f}^{\ast}(\mathbf{r}^{k})=\mathbf{g}(\mathbf{r}^{k})$, $k=1,\ldots,K$. 

Unfortunately, if the only information about $\mathbf{f}^{\ast}$ are the observations in $\mathcal{D}$ and the fact that  $\mathbf{f}^{\ast}\in C(\mathcal{X})$, the worst-case error can be arbitrarily large for some appropriately chosen $\mathbf{g}\in C(\mathcal{X})$. 
However, if $\mathbf{f}^{\ast}$ belongs to a compact subset of $C(\mathcal{X})$, the $\sup$ in \eqref{eq:error} is attained and we can guarantee a worst-case finite error. A sufficient condition for compactness of a subset in the space $\mathcal{C}(\mathcal{X})$ is that all functions in the subset are Lipschitz continuous with the same Lipschitz constant. Moreover, Lipschitz continuity imposes a nonlinear restriction on the function class. In this case it has been shown in \cite{Weinberger1959}, that for any given $\mathbf{r}\in\mathcal{X}$, the values $\mathbf{f}^{\ast}(\mathbf{r})$ belong to a closed interval, and the optimal approximation to $\mathbf{f}^{\ast}(\mathbf{r})$ is the midpoint of this interval. This means that no matter how inconvenient the machine learning scenario is (for example, a small sample set and fast changing statistics), we are guaranteed a certain finite worst-case error. 
Therefore, in addition to the monotonicity and uniform continuity properties of the fixed point function we have proved in the previous section, we make the following assumption\footnote{The function $\mathbf{f}:\mathcal{X}\rightarrow\mathcal{Y}$ can indeed be shown to be Lipschitz continuous by applying the theory of implicit functions, and by exploiting the fact that $\mathbf{f}:\mathcal{X}\rightarrow\mathcal{Y}$ is continuously differentiable over the compact set $\mathcal{X}$. The proof has technical details which we omit for brevity.}:
\begin{assumption}
The function $\mathbf{f}:\mathcal{X}\rightarrow\mathcal{Y}$ is a (component-wise) Lipschitz monotone function (LIMF) on the set $\mathcal{X}$ (see Definition \ref{def:lmf}).
\label{assumption:one}
\end{assumption}
We can now state the optimal approximation as an optimization problem.
\begin{defn}[\textit{Optimal Approximation}]\label{def:optimal_approximation}
\textit{Let $\mathcal{X} \subset \mathbb{R}^{N}_{++}$ and $\mathcal{Y}\subset[0,1]^M$. Let $\mathcal{D}=\{(\mathbf{x}^k,\boldsymbol{\rho}^k) \in (\mathcal{X} \times \mathcal{Y})\}_{k=1}^{K}$, be a data set and assume that $\boldsymbol{\rho}^k:=\mathbf{f}(\mathbf{x}^k), k=1,\ldots,K$, are values generated by an unknown function $\mathbf{f}\in \mathcal{F}:\mathcal{X}\rightarrow\mathcal{Y}$, where $\mathcal{F}\subset C(\mathcal{X})$ is a set of LIMF functions. The minimax optimal approximation (or optimal recovery) problem is stated as follows:
\begin{problem}\label{problem:main}\cite{Sukharev1992,Beliakov2005,Belford1972}
Find $\mathbf{g}:\mathcal{X}\rightarrow\mathcal{Y}$, such that
\begin{equation}
\mathbf{g}\in\underset{\mathbf{g}\in S}{\arg\min}\,\text{E}_{\text{max}}(\mathbf{g})
\end{equation}
where $S:=\left\{\mathbf{g}\in C(\mathcal{X})|\mathbf{g}(\mathbf{x}^k)=\mathbf{f}(\mathbf{x}^k),\forall k \in \left\{1,\ldots,K\right\}\right\}$, and $\text{E}_{\text{max}}(\mathbf{g}):=\max_{\mathbf{f}\in \mathcal{F}}\|\mathbf{f}-\mathbf{g}\|_{C(\mathcal{X})}$ is the worst-case error from \eqref{eq:error} computed over the set $\mathcal{F}$.
\end{problem}}
\end{defn}
%\begin{rem}
%Problem \ref{problem:main} is a generalization of the Chebychev minimax approximation. In the canonical form of the minimax problem, the function $\mathbf{f}^{\ast}\in C(\mathcal{X})$, where $\mathcal{X}$ is compact, is given. In our case, however, we only have the information that $\mathbf{f}^{\ast}$ lies in some compact subset $\mathcal{F}$ of $C(\mathcal{X})$. Therefore, any optimal solution $\mathbf{g}\in C(\mathcal{X})$ to Problem \ref{problem:main} is a minimax approximation, simultaneously, for each $\mathbf{f} \in \mathcal{F}$. Any optimal solution to Problem \ref{problem:main} is refered to as a \textit{Chebychev center}\cite{Sukharev1992} and its existence is guaranteed by the fact that $\mathcal{F}$ is a bounded set in the space of real-valued contuous function over a compact set. 
%\end{rem}
In \cite{Beliakov2005}, the author provides a framework for monotone interpolation of Lipschitz functions defined over a compact set by using a \textit{central scheme} \cite{Sukharev1992,traub1980}, that can be used to construct an optimal solution to Problem \ref{problem:main}.  
%The name \textit{central scheme} refers to the fact that optimal solutions to the Problem \ref{problem:main} are the Chebychev centres of 
%Therefore, under Assumption \ref{assumption:one} and compactness of the feasible rate demand set, we can use this framework to construct the minimax optimal approximation of Problem \ref{problem:main}. 
The following Fact summarizes the important properties of an optimal solution constructed using this framework. 
\begin{fact}
Let $\mathcal{D}=\{(\mathbf{x}^k,\boldsymbol{\rho}^{k}) \in (\mathcal{X} \times \mathcal{Y})\}_{k=1}^{K}$, be a dataset generated by an unknown function $\mathbf{f} \in \mathcal{F}$, where $\mathcal{F}$ is a set of LIMF functions (see Definition \ref{def:lmf}). Then, we have the following:
\begin{enumerate}
\item An optimal minimax approximation $\mathbf{g}\in C({\mathcal{X}})$ of $\mathbf{f}\in\mathcal{F}$ is given by 
\begin{equation}\label{eqn:opt_inter}
 (\forall i \in \mathcal{M})(\forall \mathbf{x} \in \mathcal{X}) \,\, g_{i}(\mathbf{x})=\frac{\sigma^{i}_{l}(\mathbf{x})+ \sigma^{i}_{u}(\mathbf{x})}{2},
\end{equation}
where $\sigma^{i}_l(\mathbf{x})=\max_k\{f_{i}(\mathbf{x}^{k})-L_{i}\|(\mathbf{x}-\mathbf{x}^{k})_{+}\|\}$, $\sigma^{i}_u(\mathbf{x})=\min_k\{f_{i}(\mathbf{x}^{k})+L_{i}\|(\mathbf{x}-\mathbf{x}^{k})_{+}\|\}$, $f_{i}(\mathbf{x}^{k})=\rho_{i}^{k}$, and $L_{i}$ is the Lipschitz constant of the $i$th component $f_{i}$. 
\item \label{fact_two_three}  $\mathbf{g} \in \mathcal{F}\subset C(\mathcal{X})$.
\end{enumerate}
\label{fact:main}
\end{fact}

%%%%%%%%%%%%%%%%%%%%%%%%%%%%%%%%%%%%%%%%%%%%%%%%%%%%%%%
\section{Algorithm and Simulation}
%In this section, we compare the performance of the proposed method with two standard techniques in machine learning. The aim of the simulation is to compare the performance under the practical but challenging machine learning scenario whera a small number of training samples $K$ are available the dimension of the feature vector $N$ is large. 
We consider a neighborhood with $M=10$ BS sites and $N=50$ TPs placed randomly. Each TP is connected to a single BS with the best received SNR. The pathloss for links between BSs and TPs follows the 3GPP ITU propagation model for urban macro cell  environments. In the following, we restrict our attention to a single BS and omit the index $i$ because the cell-load approximation ($g_i$ in \eqref{eqn:opt_inter}) is computed independently at each BS.
\subsection{Noisy Training Data}\label{sec:noisy_measurements}
Practical systems are subject to noise during measurement, so that instead of a data set $\mathcal{D}=\{(\mathbf{r}^k,f(\mathbf{r}^{k}))\}_{k=1}^{K}$, a noisy training data set $\mathcal{D}^{\text{noise}}=\{(\mathbf{r}^k,y^{k}=f(\mathbf{r}^{k}))+\epsilon(\mathbf{r}^k))\}$ is available, where $\epsilon(\mathbf{r}^k)$ is the measurement noise assumed to be bounded. 
As a consequence, $y^{k}$ might not be compatible with the monotonicity property of $f$ and must be smoothed to obtain a compatible set. In more detail, we first estimate the Lipschitz constant $L$ by $L:=\max_{k\neq j} \frac{|y^{k}-y^{j}|-2 \epsilon}{\|\mathbf{r}^{k}-\mathbf{r}^{j}\|}$, where $\epsilon:=\sup_k |\epsilon(\mathbf{r}^k)|$ \cite{Calliess2014}. 
The monotone-smoothing problem is given by a linear program (LP) \cite{Beliakov2005}
\begin{align}
\centering
\underset{q_{+}^k,q_{-}^k \geq 0}{\min.} & \,\sum^{K}_{k=1}{|q^k|}\nonumber\\
\text{s.t.}\,\, & q^k-q^j\leq y^{j}-y^{k} + L\|(\mathbf{r}^{k}-\mathbf{r}^{j})_{+}\|,\nonumber\\
&\forall k,j\in \{1,2,\ldots,K\}\nonumber
\end{align}
 where $q^k=q_{+}^k-q_{-}^k$, $|q^k|=q_{+}^k+q_{-}^k$, and $q_{+}^k,q_{-}^k \geq 0$ are the optimization variables. The smoothed compatible values can be calculated as $\rho^{k}:=y^{k}+q^k$. An LP is a convex optimization problem and can be solved by a standard convex or LP solver. 
\begin{algorithm}[t]
 \small
 \begin{algorithmic}[1]\label{algorithm:one}
  \caption{Load Estimation At Each BS}
\State \textbf{Training:} (a) Collect the training data set $\mathcal{D}^{\text{noise}}=\{(\mathbf{r}^k,y^{k}=f(\mathbf{r}^{k})+\epsilon(\mathbf{r}^k))\}_{k=1}^{K}$. (b) Perform the estimation of $L$ and data smoothing to construct the compatible set $\mathcal{D}=\{(\mathbf{r}^k,\rho^{k})\}_{k=1}^{K}$ as described in Section \ref{sec:noisy_measurements}.  
\State \textbf{Online Prediction:} Given a new rate demand vector $\mathbf{r}\in\mathcal{X}$, perform the following direct computation in Fact \ref{fact:main}:
        \begin{equation}
	       \begin{split}
	     g(\mathbf{r})=\frac{1}{2}(\underset{k}{\max}\{\rho^{k}-L\|(\mathbf{r}-\mathbf{r}^{k})_{+} \| \})+\\
			\frac{1}{2}(\underset{k}{\min} \{\rho^{k}+L \|(\mathbf{r}-\mathbf{r}^{k})_{+}\| \}).\nonumber
	        \end{split}
	        %\label{eqn:optimal_approximation}
         \end{equation}	
\end{algorithmic}
\end{algorithm}

\subsection{Implementation and Complexity}
The cell-load estimation algorithm is shown in Algorithm $1$. Note that, for a given $\mathbf{r}\in\mathcal{X}$ each BS $i\in\mathcal{M}$ can calculate the component $g_{i}(\mathbf{r})$ independently of other BSs using \eqref{eqn:opt_inter}. Therefore, Algorithm 1 is scalable to a larger dense network and is amenable to distributed implementation. The \textit{training} step can be performed by standard convex solvers whereas the complexity of the \textit{online prediction} step is linear in sample size $K$, i.e, $\mathcal{O}(K)$. Therefore for small sample sizes considered in this study, Algorithm $1$ exhibits a fast computational speed. 
\subsection{Results}
We train the network over the set $\mathcal{X}=[\mathbf{r}^{\text{min}},\mathbf{r}^{\text{max}}] \subset \mathbb{R}_{++}^{50}$, where $\mathbf{r}^{\text{min}}\geq (10^{6})\,\boldsymbol{1}$ and $\mathbf{r}^{\text{max}}\leq (10^{7})\,\boldsymbol{1}$ is the pre-configured feasible range (in bits/s) of rate vectors, where $\boldsymbol{1}\in \mathbb{R}_{++}^{50}$ is the ones vector. We calculate the cell-load values using the fixed point iterative method in Fact $1$ with the cell-load mapping \eqref{eqn:load_mapping}. Other important parameters are: $RB=20\,\text{MHz}$, $p_i=1\,\text{W}$, $\sigma^{2}=1.38\times10^-23\times300/20\times10^5$. Normally distributed random noise with $\epsilon=0.05$ is added to the data. The \textit{training step} is performed by a standard convex solver. 

We compare the performance of Algorithm $1$ and two other standard machine learning techniques, namely the standard Gaussian kernel regression and the 2-nearest neighbor interpolation. Note that neither of these two techniques are in general shape preserving. We use these two techniques because they are able to handle problems involving high-dimensional multivariate scattered data such as the case in this study \cite{Beliakov2005}. For brevity we compare the quality (in terms of \textit{Pearson's correlation coefficient}) and accuracy (in terms of \textit{root mean square error} (RMSE)) for cell-load predictions at a single BS. Similar results were obtained for each BS. We simulate increasing sample size $K$ and make $100\,000$ test predictions at random values of rate demand vectors in $\mathcal{X}=[\mathbf{r}^{\text{min}},\mathbf{r}^{\text{max}}]$ for each value of $K$ to gather reliable statistics. 

It can be  observed in Figure \ref{fig:figure_one} that our proposed framework shows a more robust and consistent performance both in terms of quality of prediction and accuracy over the range of sample sizes as compared to the other two techniques, especially for small sample sizes, i.e, it is robust under uncertainty. The improvement in RMSE with increasing sample sizes is due to the decrease in uncertainty about the true function $\mathbf{f}$. Note that even though the cost function \eqref{eq:error} which Algorithm $1$ optimizes is not the same as the RMSE, we can still represent its performance using a standard error measure like RMSE. At values near $K=600$ the three techniques show comparable performance in terms of RMSE, but in contrast to our framework, the other techniques are not guaranteed to be shape preserving and the predictions might not be compatible with the monotonicity property of the function. 

\begin{figure}[t]
  \centering
 \includegraphics[width=0.49\textwidth]{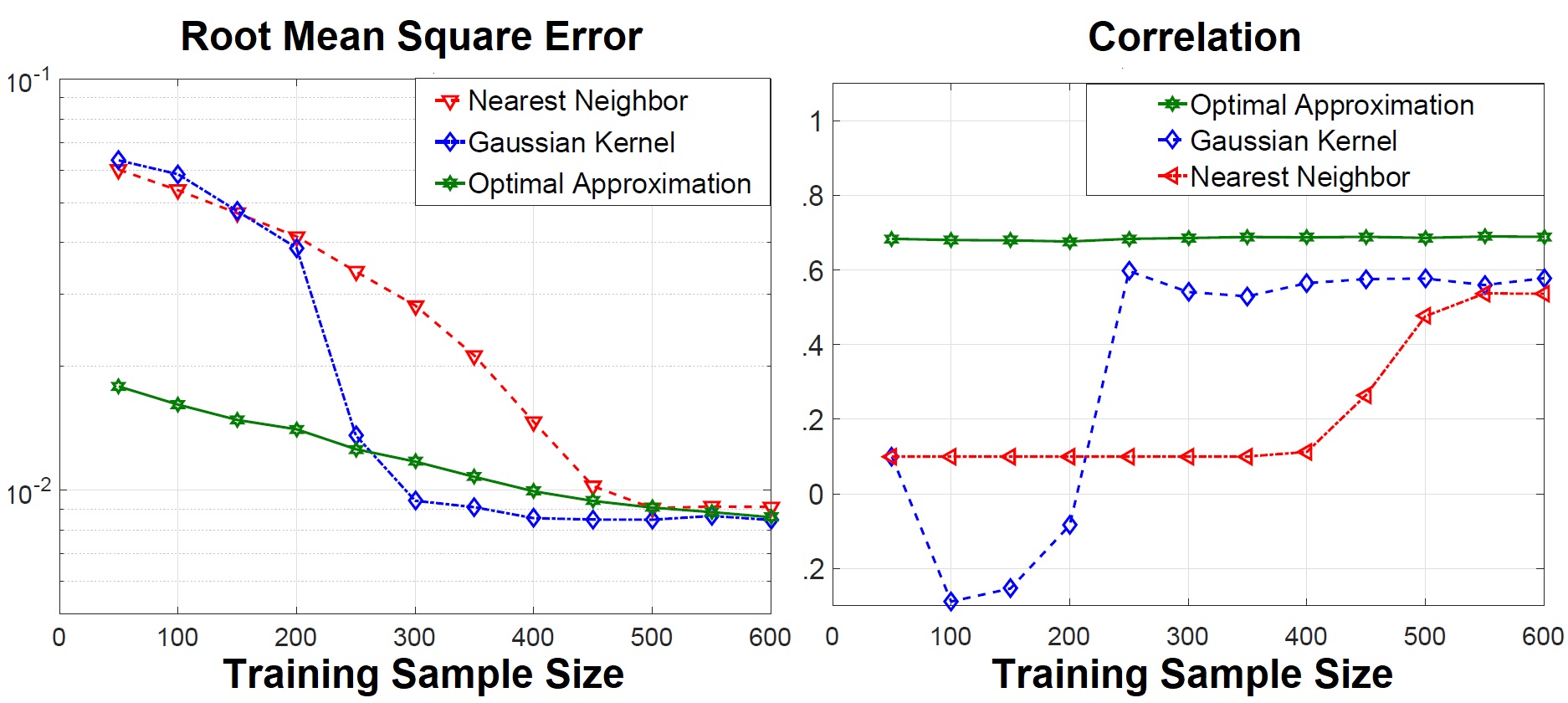}
 \caption{RMSE and (Pearson's) Correlation coefficient}
\label{fig:figure_one}
\end{figure}

\newpage
\bibliographystyle{IEEEbib}
\bibliography{IEEEabrv,library}

\begin{thebibliography}{10}

\bibitem{Siomina2012}
I.~Siomina,
\newblock ``{Analysis of Cell Load Coupling for LTE Network Planning and
  Optimization},''
\newblock {\em IEEE Transactions on Wireless Communications}, vol. 11, no. 6,
  pp. 2287--2297, June 2012.

\bibitem{Fehske2012}
A.~J. Fehske and G.~P. Fettweis,
\newblock ``Aggregation of variables in load models for interference-coupled
  cellular data networks,''
\newblock in {\em 2012 IEEE International Conference on Communications (ICC)},
  June 2012, pp. 5102--5107.

\bibitem{Majewski2010}
K.~Majewski and M.~Koonert,
\newblock ``Conservative cell load approximation for radio networks with
  shannon channels and its application to lte network planning,''
\newblock in {\em 2010 Sixth Advanced International Conference on
  Telecommunications}, May 2010, pp. 219--225.

\bibitem{Ho2014}
Chin~Keong Ho, Di~Yuan, and Sumei Sun,
\newblock ``Data offloading in load coupled networks: A utility maximization
  framework,''
\newblock {\em IEEE Transactions on Wireless Communications}, vol. 13, no. 4,
  pp. 1921--1931, 2014.

\bibitem{miguel2016}
M.~A. Gutierrez-Estevez, R.~L.~G. Cavalcante, S.~Stanczak, J.~Zhang, and
  H.~Zhuang,
\newblock ``A distributed solution for proportional fairness optimization in
  load coupled ofdma networks,''
\newblock in {\em IEEE Global Conference on Signal and Information Processing
  (GlobalSIP' 16)}, 2016.

\bibitem{daniyal2016}
D.A. Awan, R.~L.~G. Cavalcante, and Slawomir Stanczak,
\newblock ``Distributed ran and backhaul optimization for energy efficient
  wireless networks,''
\newblock in {\em IEEE Global Conference on Signal and Information Processing
  (GlobalSIP' 16)}.

\bibitem{Pollakis2016}
E.~Pollakis, Renato L~G Cavalcante, and S.~Stanczak,
\newblock ``Traffic demand-aware topology control for enhanced
  energy-efficiency of cellular networks,''
\newblock {\em EURASIP Journal on Wireless Communications and Networking}, vol.
  2016, no. 1, pp. 1--17, 2016.

\bibitem{siomina2012b}
I.~Siomina and D.~Yuan,
\newblock ``Load balancing in heterogeneous lte: Range optimization via cell
  offset and load-coupling characterization,''
\newblock in {\em 2012 IEEE International Conference on Communications (ICC)},
  June 2012, pp. 1357--1361.

\bibitem{Wong04}
I.~C. Wong, Zukang Shen, B.~L. Evans, and J.~G. Andrews,
\newblock ``A low complexity algorithm for proportional resource allocation in
  ofdma systems,''
\newblock in {\em IEEE Workshop onSignal Processing Systems, 2004. SIPS 2004.},
  Oct 2004, pp. 1--6.

\bibitem{siomina2015}
I.~Siomina and D.~Yuan,
\newblock ``Optimizing small-cell range in heterogeneous and load-coupled lte
  networks,''
\newblock {\em IEEE Transactions on Vehicular Technology}, vol. 64, no. 5, pp.
  2169--2174, May 2015.

\bibitem{Beliakov2005}
G.~Beliakov,
\newblock ``Monotonicity preserving approximation of multivariate scattered
  data,''
\newblock {\em BIT Numerical Mathematics}, vol. 45, no. 4, pp. 653--677, 2005.

\bibitem{Kotlowski16}
Wojciech Kotlowski,
\newblock ``Online isotonic regression,''
\newblock in {\em Proceedings of the 29th Conference on Learning Theory, {COLT}
  2016, New York, USA, June 23-26, 2016}, 2016, pp. 1165--1189.

\bibitem{Cavalcante2016}
R.~L.~G. Cavalcante, Y.~Shen, and S.~Stanczak,
\newblock ``Elementary properties of positive concave mappings with
  applications to network planning and optimization,''
\newblock {\em IEEE Transactions on Signal Processing}, vol. 64, no. 7, pp.
  1774--1783, April 2016.

\bibitem{yates95}
R.~D. Yates,
\newblock ``A framework for uplink power control in cellular radio systems,''
\newblock {\em IEEE Journal on Selected Areas in Communications}, vol. 13, no.
  7, pp. 1341--1347, Sep 1995.

\bibitem{Golub1971}
G.~H. Golub and L.~B. Smith,
\newblock ``Algorithm 414: Chebyshev approximation of continuous functions by a
  chebyshev system of functions,''
\newblock {\em Commun. ACM}, vol. 14, no. 11, pp. 737--746, Nov. 1971.

\bibitem{krause01}
U.~Krause,
\newblock ``Concave perron–frobenius theory and applications,''
\newblock {\em Nonlinear Analysis: Theory, Methods \& Applications}, vol. 47,
  no. 3, pp. 1457 -- 1466, 2001.

\bibitem{renato20172}
R.~L.~G. {Cavalcante} and S.~{Stanczak},
\newblock ``{The role of asymptotic functions in network optimization and
  feasibility studies},''
\newblock in {\em IEEE Global Conference on Signal and Information Processing
  (GlobalSIP), to appear}, Nov. 2017.

\bibitem{nuzman07}
C.~J. Nuzman,
\newblock ``Contraction approach to power control, with non-monotonic
  applications,''
\newblock in {\em IEEE GLOBECOM 2007 - IEEE Global Telecommunications
  Conference}, Nov 2007, pp. 5283--5287.

\bibitem{Weinberger1959}
M~Golomb and H.F. Weinberger,
\newblock {\em On Numerical Approximation},
\newblock The University of Wisconsin Press, Madison, r.e. langer (ed.)
  edition, 1959.

\bibitem{Sukharev1992}
Aleksei~G. Sukharev,
\newblock {\em Minimax Models in the Theory of Numerical Methods},
\newblock Kluwer Academic Publishers, Norwell, MA, USA, 1992.

\bibitem{traub1980}
J.F. Traub and H.~Wo{\'z}niakowski,
\newblock {\em A general theory of optimal algorithms},
\newblock ACM monograph series. Academic Press, 1980.

\bibitem{Belford1972}
Geneva~G. Belford,
\newblock ``Uniform approximation of vector-valued functions with a
  constraint,''
\newblock {\em Mathematics of Computation}, vol. 26, no. 118, pp. 487--492,
  1972.

\bibitem{Calliess2014}
Jan-Peter Calliess,
\newblock {\em Conservative decision-making and inference in uncertain
  dynamical systems},
\newblock Ph.D. thesis, Department of Engineering Science, University of
  Oxford, 2014.

\end{thebibliography}

\end{document}